\documentclass[a4paper,12pt]{article}

\title{On the hull-variation problem of equivalent vector rank-metric codes}
\author{Duy Ho and Trygve Johnsen \\
Department of Mathematics and Statistics \\
UiT The Arctic University of Norway, Tromsø
9037, Norway\\
Email: duyho92@gmail.com, trygve.johnsen@uit.no}

\date{ } 
\usepackage{todonotes}
\usepackage{tikz-cd}
\usepackage{amsthm,amsmath,amssymb} 
 
\usepackage{url} 
\usepackage{booktabs}
\usepackage{enumitem}
\allowdisplaybreaks
\setenumerate[0]{label=(\roman*)}
\begin{document}

\maketitle

\theoremstyle{plain} 
\newtheorem{lemma}{Lemma}[section] 
\newtheorem{theorem}[lemma]{Theorem}
\newtheorem{corollary}[lemma]{Corollary}
\newtheorem{proposition}[lemma]{Proposition}
\newtheorem{question}{Open problem}[section]

\theoremstyle{definition} 
\newtheorem{definition}{Definition}[section] 
\newtheorem{remark}{Remark}[section] 
\newtheorem{example}{Example}

\newcommand{\eps}{\varepsilon}
\newcommand{\inprod}[1]{\left\langle #1 \right\rangle}
\newcommand{\la}{\lambda} 
\newcommand{\al}{\alpha}
\newcommand{\om}{\omega} 
\newcommand{\gam}{\gamma}
\newcommand{\be}{\beta}
\newcommand{\sig}{\sigma}
\newcommand\rank{\mathrm{rank}}

\newcommand{\F}{\mathbb{F}}
\newcommand{\K}{\mathbb{K}}
\newcommand{\E}{\mathcal{E}}
\newcommand{\M}{\mathcal{M}}

\newcommand{\duy}[1]{\todo[inline,color=purple!100]{D: #1}}
\newcommand{\trygve}[1]{\todo[inline,color=blue!100]{J: #1}}

\begin{abstract}  
The intersection of a linear code with its dual is called the hull of the code. 
It is known that, for classical linear codes under the Hamming-metric, the dimension of the hull can be reduced up to equivalence. 
This phenomenon leads to the so-called hull-variation problem formulated by Hao Chen in 2023.
In this paper, we consider the analogous problem for vector rank-metric codes, along with their associated matrix codes and extended block codes. 
Our results include the fact that every vector rank-metric code over any finite field  $\mathbb{F}_q$, in particular when  $q=2$ or $q=3$, is equivalent to an LCD code.
\end{abstract}



\section{Introduction}

A classical linear code \( C \) is a subspace of \( \mathbb{F}_q^n \) with the Hamming distance used as the metric.  
The hull of a classical linear code is defined as the intersection of the code with its dual code.  
The concept of the hull of classical linear codes has gained significant interest in recent years, particularly in applications such as quantum error correction and cryptography.

Classical linear codes with trivial hulls are known as linear codes with complementary duals (LCD codes). These codes were first introduced by Massey in \cite{massey}.
 In \cite{sendrier1997},  Sendrier proved that LCD codes are asymptotically good and used them in relation to equivalence testing of linear codes in \cite{sendrier2000}. 
 Recently, LCD codes have become an attractive topic of research as they offer solutions to many cryptographic problems, for example against side-channel attacks and fault non-invasive attacks, see \cite{carlet2016}.  
 In \cite{CMTQP}, it was shown that  any linear code over $\mathbb{F}_q$ (for $q  > 3$) is equivalent to a Euclidean LCD code and any linear code over $\mathbb{F}_{q^2}$ (for $q  > 2$) is equivalent to a Hermitian LCD code.

 The latter result was generalized recently  in \cite{chen2023-3}  and \cite{ling2024}, where it was shown that for any classical linear code $C$ over $\mathbb{F}_{q^2}$ ($q  > 2$) and with Hermitian hull dimension $h$, one can find  a new code which is equivalent to $C$ with Hermitian hull dimension $\ell$ for any $\ell$ with $0 \le \ell\le h$. This phenomenon leads to the so-called
hull-variation problem formulated by Hao Chen in 2023, see \cite{chen2023}. The problem is stated as follows. 

\bigskip
 \textit{Hull-Variation Problem:} When a linear code \( C \) is transformed to an equivalent linear code $C'$, how is its Euclidean or Hermitian hull changed?
\bigskip

 Apart from the Hamming distance, various distance metrics are used in coding theory for different applications. 
One notable example that we will consider here is the rank-metric.
 The class of rank-metric codes was first considered by Delsarte \cite{delsarte1978} in the late 1970s. In 1985, Gabidulin~\cite{gabidulin1985} contributed significantly to the development
of error-correcting codes in the rank-metric, establishing several important properties of rank-metric codes.
 In the past decade, rank-metric codes have become an important research topic with various applications, especially in network coding, cryptography, and error correction in random linear network coding.   For  recent surveys  on the theory and applications of rank-metric code, we refer the reader to \cite{bartz2022}  and \cite{gorla2021}. For a description of applications of vector rank-metric codes, where it is important that the codes are LCD, see \cite[pp. 107--110]{kssd}.

% Associated with matrix codes are $q$-polymatroids. There are many invariants

 Building on the previous results in the literature, in this paper we  consider the (Euclidean) hull-variation problem of equivalent rank-metric codes,  along with their associated matrix codes and extended block codes. 
 We show that for any prime power $q$  and any vector rank-metric code with a non-trivial hull, there is an equivalent code whose hull dimension is smaller. As a consequence, we show that every vector rank-metric code over any finite field  $\mathbb{F}_q$, including the cases  $q=2$ and $q=3$, is equivalent to an LCD code.
 These results demonstrate that the behavior of   hulls in the context of vector rank-metric codes is different from that of classical Hamming-metric codes.   Our results also imply that the hull dimension of a matrix code is not an invariant captured by the associated $(q,m)$-polymatroids.

 The content of the paper is organized as follows.  In Section 2, we recall preliminary results for linear codes. In Section 3, we consider the hull-variation problem of vector rank-metric and provide examples. In Section 4, we consider the hull-variation problem of the associated codes and implications in the context of  $(q,m)$-polymatroids.

\section{Preliminaries}
% Let $u = (u_1,\dots, u_n)$ and $v = (v_1,\dots,v_n)$ in $\mathbb{F}^n_{q^m}$.

We follow some material from \cite{bonini2023}, \cite{galvez2020}, \cite{gorla2019}, \cite{morrison2014}, \cite{ravagnani2016}, and \cite{reijnders2024}.
Let  $q$   be a   prime power, and let $m, n, k$  be positive integers.
Let $\mathbb{F}_q$ be the finite field with $q$ elements.
Let $\mathbb{F}_q^{n \times m}$ be  the \( \mathbb{F}_q \)-vector space of matrices of size \( n \times m \) with entries in \( \mathbb{F}_q \).

\subsection{Hamming-metric codes, vector rank-metric codes and matrix codes}

An \textit{$\mathbb{F}_{q^m}$-linear  code} $C$ of length $n$ is a subspace of $\mathbb{F}_{q^m}^n$ over $\mathbb{F}_{q^m}$. Let $k$ be the  dimension of $C$ over $\mathbb{F}_{q^m}$. The code $C$   is often described in terms of a \textit{generator matrix}  \( G \in \mathbb{F}_{q^m}^{k \times n} \), which is a full-rank matrix whose rows generate \( C \). 
The \textit{dual} of $C$ is the $\mathbb{F}_{q^m}$-linear code
$
C^\perp := \{v \in  \mathbb{F}^n_{q^m}  \mid \langle v,w \rangle =0 \text{ for all $w \in C$}\},   
$
where $\langle \cdot,\cdot \rangle$ is the standard inner product of $\mathbb{F}^n_{q^m}$.

Let $u=(u_1, \dots, u_n)$ and $v=(v_1, \dots, v_n)$ be in $\mathbb{F}_{q^m}^n$. The \emph{Hamming distance} between $u,v\in \F_{q^m}^n$ is
$d_H(u,v)=|\{i:u_i\ne v_i\}|$, while the \emph{rank distance} is
\[
d_R(u,v)=\dim_{\F_q}(\langle u_1-v_1,\dots,u_n-v_n\rangle_{\F_q}).
\]
If $C$ is equipped with the Hamming distance, it is called a \textit{Hamming-metric code} and denoted by $[n,k]_{q^m}$; if it is equipped with the rank distance, it is called a \textit{vector rank-metric code} and denoted by $[n,k]_{q^m/q}$.

A \textit{matrix code} of size \( n \times m \) over \( \mathbb{F}_q \) is an \( \mathbb{F}_q \)-linear subspace  \( C \subseteq \mathbb{F}_q^{n \times m} \).  
If $C$ is of dimension $k$, then $C$ is called an $[n \times m ,k]_q$ matrix code over \( \mathbb{F}_q \).
 Given a matrix $M \in \mathbb{F}_q^{n \times m}$, we write $\text{Tr}(M)$ for the trace of $M$,  and $M^\top$ for the transpose of $M$.
We define the \textit{trace product} of matrices \( M, N \in \mathbb{F}_q^{n \times m} \) as
$
\langle M, N \rangle := \text{Tr}(M N^\top), 
$
giving the dual code
\[
C^\perp := \{ N \in \mathbb{F}_q^{n \times m} \mid \langle M, N \rangle = 0 \text{ for all } M \in C \}.
\] 
The \textit{rank weight} of a matrix \( M \in C \) is the rank of  $M$ and is denoted \( \operatorname{wt}_R(M) \).   The \textit{rank distance} between two matrices \( M, N \in C \) is the rank of their difference \( M-N \), that is, 
\[
   d_R(M,N) = \operatorname{wt}_R(M - N).
\]

\subsection{Codes associated with vector rank-metric codes}

 % In order to easily handle matrix codes, we transform matrix codes to linear block codes. 
 For a matrix \( X = [x_{ij}] \in \mathbb{F}_q^{n \times m} \), the \textit{extended row vector} corresponding to \( X \) is the vector  
\[
\rho(X) = (x_{11}, x_{12}, \dots, x_{1m}, x_{21}, x_{22}, \dots, x_{2m}, \dots, x_{n1}, x_{n2}, \dots, x_{nm}) \in \mathbb{F}_q^{nm},
\]
formed by concatenating the rows of the \( n \times m \) matrix \( X \).  The ``extended row vector” is also known in the
literature  as ``flattening” or ``vectorization” of a matrix.
If \( C \subseteq \mathbb{F}_q^{n \times m} \) is a matrix code, then the \textit{extended block code} of \( C \) is given by  
\[
\rho(C) = \{ \rho(X) \mid X \in C \} \subseteq \mathbb{F}_q^{nm}.
\]
% \end{definition}
The map $\rho$ defines a natural correspondence between matrix codes and linear block codes. Moreover, there is also a natural correspondence between the duals of matrix codes and the
 duals of extended block codes, that is,
\begin{equation} \label{rhocom}
\rho(C^\perp) = \rho(C)^\perp.
\end{equation} 
A generator matrix \( G \) of \( \rho(C) \) will also be called a \textit{generator matrix} of \( C \).  
Each row of \( G \) is the image of the generators of the matrix code \( C \) under \( \rho \).  
This generator matrix is of size \( k \times nm \).
We note that  a generator matrix does not determine a unique matrix code and  it is important to indicate the size of the matrix code whenever a generator matrix is mentioned, see \cite{galvez2020}.

Let \( \mathcal{G} = \{\gamma_1, \dots, \gamma_m\} \) be a basis of \( \mathbb{F}_{q^m} \) over \( \mathbb{F}_q \). 
The matrix associated with a vector \( \alpha \in \mathbb{F}_{q^m}^n \) with respect to \( \mathcal{G} \) is the \( n\times m \) matrix \( M_\mathcal{G} (\alpha) \) with entries in \( \mathbb{F}_q \)  defined by  

\[
\alpha_i = \sum_{j=1}^{m} M_\mathcal{G} (\alpha)_{ij} \gamma_j \quad \text{for all } i = 1, \dots, n.
\] 
With the  \textit{matrix code associated with a vector rank-metric code \( C \subseteq \mathbb{F}_{q^m}^n \) with respect to the basis \( \mathcal{G} \)} we will here mean   
\[
\mathcal{C}_\mathcal{G}(C) = \{ M_\mathcal{G}(\alpha) : \alpha \in C \} \subseteq  \mathbb{F}_q^{n \times m}.
\]
 Let $\mathrm{Trace}:\F_{q^m}\to\F_q$ be the field trace.
Two bases $G=\{\gamma_i\}$ and $G'=\{\gamma_i'\}$ of $\F_{q^m}$ over $\F_q$ are
\emph{dual} if $\mathrm{Trace}(\gamma_i\gamma_j')=\delta_{i,j}$.  A basis is \textit{self-dual} if it is dual with itself. 
\begin{lemma} \label{ravagnani}
Let \( C \subseteq \mathbb{F}_{q^m}^n \) be a vector rank-metric code, and let \( \mathcal{G}, \mathcal{G}' \) be dual bases of $\mathbb{F}_{q^m}$ over \( \mathbb{F}_q \). Then
$
C_{\mathcal{G}'} (C^\perp) = C_\mathcal{G} (C)^\perp.
$ 
\end{lemma} 
As described in the previous subsection, for each matrix code, we have a corresponding extended block code from the map $\rho$. The  \textit{extended block code associated with a vector rank-metric code \( C \subseteq \mathbb{F}_{q^m}^n \) with respect to the basis \( \mathcal{G} \)} is   
\[
\rho(\mathcal{C}_\mathcal{G}(C)) = \{ \rho(M_\mathcal{G}(\alpha)) : \alpha \in C \} \subseteq  \mathbb{F}_q^{nm}.
\]

\subsection{Equivalent codes}

For any distance function $f$ on a $K$-vector space $V$ a linear isometry  is a $K$-linear map  
$
\varphi: V \to V
$
such that  
$
f(\varphi(v),\varphi(w)) = f(u,v)$  for every  $v,w \in V.
$
We now describe different notions of equivalent codes. 

%10. p.5 Before stating all the linear isometries in the different settings, you might want to introduce the notion of a linear isometry for any distance function $f$ and vector space $V$.

\bigskip

1. \underline{Hamming-metric codes}.  Two $\mathbb{F}_{q^m}$-linear codes $C$ and $C'$ are equivalent with respect to the Hamming-metric if there exists a monomial matrix $M \in \mathbb{F}_{q^m}^{n \times n}$ and a pair of generator matrices $G,G'$ of $C,C'$, respectively, satisfying $G'= GM$. 

\bigskip

2. \underline{Vector rank-metric codes}.  Two $[n,k]_{q^m/q}$ vector rank-metric codes \( C \) and \( C' \) are equivalent with respect to the rank-metric if there exists \( A \in \text{GL}_n(q) \) and a pair of generator matrices \( G, G' \) of \( C, C' \), respectively, satisfying  
$
G' = G A.
$
% Here, for rank-metric codes, we have chosen the definition from \cite[Definition 2]{PR}.

\bigskip

3. \underline{Matrix codes}. 
%An \textit{\( \mathbb{F}_q \)-linear isometry} of \( \mathbb{F}_q^{n \times m} \) is an \( \mathbb{F}_q \)-linear homomorphism  
%$
%\varphi: \mathbb{F}_q^{n \times m} \to \mathbb{F}_q^{n \times m}
%$
%such that  
%$
%\text{rank}(\varphi(M)) = \text{rank}(M)$  for every  $M \in \mathbb{F}_q^{n \times m}.
%$
Two matrix codes \( C, C' \subseteq \mathbb{F}_q^{n \times m} \) are equivalent if there is an \( \mathbb{F}_q \)-linear isometry  (for the rank metric)
$
\varphi: \mathbb{F}_q^{n \times m} \to \mathbb{F}_q^{n \times m}
$
such that \( \varphi(C) = C' \).

\bigskip

4. \underline{Extended block codes}.  We follow the setting in \cite{morrison2014} and \cite{morrison2015}. 
% Since the generator matrix of a  matrix code \( C  \in \mathbb{F}_q^{n \times m} \) is expressed as the generator  matrix of its extended block code \( \rho(C) \), we can translate the \( \mathbb{F}_q \)-linear isometries  described  above as maps acting on the generator matrix of \( C \). In particular,  these
%  maps are known  as   
% right multiplication by  elements of \( \text{GL}_{n m}(\mathbb{F}_q) \) acting on the extended row vectors.
We say that $f: \mathbb{F}_q^{nm} \to \mathbb{F}_q^{nm}$ is a \textit{linear $n \times m$ matrix-equivalence map on extended row vectors} if there exists an \( \mathbb{F}_q \)-linear isometry $g: \mathbb{F}_q^{n \times m} \to \mathbb{F}_q^{n \times m}$ such that, for all $A \in \mathbb{F}_q^{n \times m}$,
$
f(\rho(A)) = \rho(g(A)).
$  
The group of all such maps is denoted \({Equiv}_{\text{Vec}}(\mathbb{F}_{q}^{n \times m})\) and is fully described in \cite{morrison2014}.
Let $C$ and $C'$ be two matrix codes. We say that \( \rho(C) \) and \( \rho(C') \) are \textit{linear \( n \times m \) matrix-equivalent} if there exists \( F \in {Equiv}_{\text{Vec}}(\mathbb{F}_{q}^{n \times m}) \) and a pair of generator matrices \( G, G' \) of \( \rho(C) , \rho(C') \), respectively, satisfying  
$
G' = G F.
$

\bigskip

5. \underline{Codes associated with vector rank-metric codes}.  
In \cite{morrison2014}, we have the following.
 
 \begin{lemma} \label{equivalencetransfer} Let $C$ and $C'$ be two vector rank-metric codes. If $C$ and $C'$ are equivalent, then $\mathcal{C}_\mathcal{G}(C)$ and $\mathcal{C}_\mathcal{G}(C')$ are equivalent. 

\end{lemma}

Also from the setting above, we have the following. 

 \begin{lemma} \label{equivalencetransferblock} Let $C$ and $C'$ be two vector rank-metric codes. If $C$ and $C'$ are equivalent, then $\rho(\mathcal{C}_\mathcal{G}(C))$ and $\rho(\mathcal{C}_\mathcal{G}(C'))$ are linear \( n \times m \) matrix-equivalent. 

\end{lemma}

\subsection{The hull of codes}
We have defined the duality notions  for Hamming-metric codes, vector rank-metric codes, matrix codes, and extended block codes. If $C$ is one such code, then the \textit{(Euclidean) hull} of $C$ is the intersection of $C$ and $C^\perp$. We denote the hull of $C$ by $H(C)$. 
The following characterization is well-known for Hamming-metric codes and can be easily verified for vector rank-metric codes.
\begin{lemma} \label{rankhull} Let $C$ be either an $[n,k]_{q^m}$ Hamming-metric code    or an   $[n,k]_{q^m/q}$  vector rank-metric  code.  Let $G$  be a generator matrix of $C$ over $\mathbb{F}_{q^m}$. Then
$\dim(H(C))=k-\rank(GG^\top)$.
\end{lemma}
% \begin{proof}  A vector $\textbf{v}$ is an element of $H(C)$ if and only if $\textbf{v} \in C$ and $\textbf{v} \in C^\perp$. The first condition implies that $\textbf{v}= \textbf{u}G$ for some $\textbf{u} \in \mathbb{F}_{q^m}^k$. The second condition implies that $\textbf{v}G^\top= \textbf{0}$.
% Then
% \[
% H(C) = \{ \textbf{u}G \mid \textbf{u} \in \mathbb{F}_{q^m}^k \text{ and } \textbf{u}GG^\top=\textbf{0} \}.
% \]
% The proof now follows from the rank-nullity theorem.
% \end{proof}

Let \( C \) be an \( \mathbb{F}_{q^m} \)-linear code,   equipped with either the Hamming distance or the rank
 distance. If \( C \cap C^{\perp} = \{\mathbf{0}\} \), then we say that \( C \) is a \textit{linear code with complementary dual}.  We will abbreviate such a code as an  \textit{(Euclidean) LCD code}. 

We are also interested in   the hulls of matrix codes and extended block codes associated with vector rank-metric codes. For this purpose,  we assume that $\mathcal{G}$ is a self-dual basis of $\mathbb{F}_{q^m}$ over $\mathbb{F}_q$.  Such a self-dual basis exists if and only if $q$ is even or both $q$ and $m$ are odd, see \cite{jungnickel1990, seroussi1980}. In this case, we have the following result from \cite{liu2019}.

 \begin{lemma} \label{hulltransfer}  Let \( q \) be even or both \( q \) and \( m \) be odd, and let  
\( \mathcal{G} \) be a self-dual basis of \( \mathbb{F}_{q^m} \) over \( \mathbb{F}_q \). Let $C$ be a $[n,k]_{q^m/q}$ vector rank-metric code. Then $\mathcal{C}_\mathcal{G}(C \cap C^\perp)=\mathcal{C}_\mathcal{G}(C) \cap \mathcal{C}_\mathcal{G}(C)^\perp$.

\end{lemma}

\section{On the hull-variation problem of  vector rank-metric codes}

It was shown in \cite[Corollary 5.15]{CMTQP} that the following holds.

\begin{theorem} \label{established} \label{carletlcd}
Let \( q > 3 \) be a prime power and let \( C \) be an \([n,k]_q\) code. Then there exists an equivalent \([n,k]_q\) Hamming-metric  code \( C' \) which is LCD.
\end{theorem}

In \cite[Corollary 2.2]{chen2023-3} and \cite[Theorem 6]{ling2024}, the previous Theorem \ref{carletlcd} was generalized in the setting of Hermitian hulls, with applications in constructing entanglement assisted quantum error-correcting codes. This result can be adapted to the case of Euclidean hulls, which we state as follows.

\begin{theorem} \label{notestablished} \label{sanlinglcd}
Let \( q > 3 \) be a prime power and let \( C \) be an \([n,k]_q\) code with \( \dim(H(C)) = h \). Then there exists an equivalent \([n,k]_q\) code \( C' \)
with \( \dim(H(C')) = \ell \) for each \( \ell\in \{0,1,\dots,h\} \).
\end{theorem}

Inspired by the proof of \cite[Theorem 6]{ling2024}, in this section, we show that the Euclidean hull dimension of vector rank-metric codes  can also be reduced under equivalence. 
We first prove a lemma. 

\begin{lemma}\label{crux}
Let $q$ be a prime power and $s \ge 1$. 
Define
\[
Z_2 =
\begin{pmatrix}
1 & 0\\[2pt]
1 & 1
\end{pmatrix},
\qquad
Z_3 =
\begin{pmatrix}
1 & 0 & 0\\[2pt]
1 & 1 & 0\\[2pt]
0 & 1 & 1
\end{pmatrix}.
\]
We define $Y \in \mathrm{GL}_s(\mathbb{F}_q)$  as follows:
\begin{enumerate}
\item If $q \in \{2,3\}$ and $s \ge 2$,  
then write $s = 2a + 3b$ with integers $a,b \ge 0$, and let
\[
Y = \operatorname{diag}(
\underbrace{Z_2, \dots, Z_2}_{a\ \mathrm{blocks}},
\underbrace{Z_3, \dots, Z_3}_{b\ \mathrm{blocks}}).
\]
\item If $q > 3$, then let
\[
Y = \operatorname{diag}(a_1, a_2, \dots, a_s),
\qquad 
a_i \in \mathbb{F}_q^{\ast}, \ \ a_i^2 \neq 1 \text{ for } 1 \le i \le s.
\]
\end{enumerate}
Then the matrix $(YY^{\top} - I_s)$ is invertible. 
\end{lemma}

\begin{proof}  The case $q>3$ can be verified easily with direct calculations, so here we will only prove the case $q \in \{2,3\}$. With the choice of $Z_2$, we   have
\[
Z_2 Z_2^\top - I_2
=
\begin{pmatrix}
0 & 1 \\
1 & 1
\end{pmatrix}.
\]
Also,
\[
Z_3Z_3^\top - I_3
=
\begin{pmatrix}
0&1&0\\
1&1&1\\
0&1&1
\end{pmatrix}.
\]
Then $Y \in GL_s(\mathbb{F}_q)$, and
\[
Y Y^\top - I_s \;=\; \operatorname{diag}\!\big(
Z_2 Z_2^\top - I_2, \dots, Z_2 Z_2^\top - I_2,\;
Z_3 Z_3^\top - I_3, \dots, Z_3 Z_3^\top - I_3
\big).
\]
Since both $Z_2 Z_2^\top - I_2$ and $Z_3 Z_3^\top - I_3$ are invertible
over $\mathbb{F}_2$ and $\mathbb{F}_3$, it follows that 
$Y Y^\top - I_s$ is invertible over $\mathbb{F}_2$ and $\mathbb{F}_3$ for all $s \ge 2$. 
\end{proof}

\begin{theorem} \label{main}
Let $q$ be a prime power and let $C$ be an $[n, k]_{q^m/q}$ vector rank-metric code with $\dim( H(C)) = h \ge 1$.
\begin{enumerate} 
\item If $q>3$, then there exists an $[n, k]_{q^m/q}$ code $C'$ equivalent to $C$ with
 $\dim(H(C')) = \ell$ for each $\ell \in \{0,1,\dots, h-1\}$. 

\item If $q \in \{2,3\}$ and $h \ge 2$, then there exists an  $[n, k]_{q^m/q}$ code $C'$ equivalent to $C$ with
 $\dim(H(C')) = \ell$ for each $\ell \in \{0,1,\dots, h-2\}$. 
 \end{enumerate}
\end{theorem}
\begin{proof} Since $H(C)$ is a  $\mathbb{F}_{q^m}$-subspace of $C$, we can regard $H(C)$ as an $[n, h]_{q^m/q}$ vector rank-metric code with generator matrix $G_0$. 
By a suitable change of basis and permutations
of the columns we get an equivalent code, where we can choose a generator
matrix $G_0$ of the form $G_0 = [\, I_h \; A \,]$.

1. Let $G$ be a generator matrix of $C$. Under suitable row operations on $G$, we can
 then assume that $G$ is of the form
\[
G \;=\;
\begin{bmatrix}
I_h & A \\[6pt]
 \mathbf{0}  & B
\end{bmatrix},
\]
where $I_h$ is the $h\times h$ identity matrix, $A\in \mathbb{F}_{q^m}^{h\times (n-h)}$,
and $B\in \mathbb{F}_{q^m}^{(k-h)\times (n-h)}$.
Let $\mathbf{g}_1, \dots, \mathbf{g}_k$ be the rows of $G$. Since $\mathbf{g}_i \in H(C)$ for $1 \le i \le h$, we have that $\langle \mathbf{g}_i,\mathbf{g}_j \rangle=0$ if $1 \le i \le h$ or $1 \le j \le h$. 
Hence, $GG^\top$ has the form
\[ GG^\top=
  \begin{bmatrix}
    \mathbf{0}  & \mathbf{0}  \\
    \mathbf{0} & BB^\top
  \end{bmatrix}.
\]
In particular, we have the following conditions
\begin{equation}\label{hullcond}
A A^\top + I_h = \mathbf{0}, 
\qquad 
A B^\top = \mathbf{0}.
\end{equation}
Also, from Lemma \ref{rankhull} we have $\text{rank}(BB^\top)= k-h$. 

% 3. Let 
% \[
% M \;=\;
% \begin{bmatrix}
% X & \mathbf{0} \\[6pt]
% \mathbf{0} & I_{n-h}
% \end{bmatrix}
% \;\in\; \mathrm{GL}_n(\mathbb{F}_q),
% \]
% where $X \in \mathrm{GL}_h(\mathbb{F}_q)$ is an invertible $h\times h$ matrix.

% Let $G'=GM.$ Then $G'$ has the form 
% \[
% G' \;=\; 
% \begin{bmatrix}
% X & A \\[6pt]
% 0 & B
% \end{bmatrix}, 
% \]
% and
% \[
% G'G'^\top=
% \begin{bmatrix}
% X X^\top + A A^\top & A B^\top\\[4pt]
% B A^\top & B B^\top
% \end{bmatrix}.
% \]
% From the conditions \eqref{hullcond}, we have
% \begin{equation} \label{hullX}
% G'G'^\top
% =\begin{bmatrix}
% X X^\top -I_h & 0\\[2pt]
% 0 & B B^\top
% \end{bmatrix}.
% \end{equation}
% Here, $\text{rank}(BB^\top)= k-h$ is fixed. We want to vary $X$ to change rank of $G'G'^\top$. 

2. Fix $\ell \in \{0, \dots, h-1\}$ if $q>3$, or $\ell \in \{0, \dots, h-2\}$ if $q\in \{2,3\}$   and $h \ge 2$. Let $Y  \in \mathrm{GL}_{h-\ell}(\mathbb{F}_q)$ be defined as in Lemma \ref{crux}. Then $YY^T-I_{h-\ell}$ is full-rank. Let 
\[
X = 
\begin{bmatrix}
Y & \mathbf{0}\\[4pt]
\mathbf{0} & I_{\ell}
\end{bmatrix}
\in \mathrm{GL}_{h}(\mathbb{F}_q),
\quad \text{and} \quad
M =
\begin{bmatrix}
X & \mathbf{0}\\[6pt]
\mathbf{0} & I_{n-h}
\end{bmatrix}
\in \mathrm{GL}_{n}(\mathbb{F}_q).
\]
Then
\begin{equation} \label{hullY}
X X^\top \;=\;
\begin{bmatrix}
Y Y^\top & \mathbf{0}\\[4pt]
\mathbf{0} & I_{\ell}
\end{bmatrix}.
\end{equation}
Let $G'=GM.$ From  conditions \eqref{hullcond} and \eqref{hullY}, we have
% \begin{equation} \label{hullX}
% G'G'^\top
% =\begin{bmatrix}
% X X^\top -I_h & 0\\[2pt]
% 0 & B B^\top
% \end{bmatrix}.
% \end{equation}

% By substituting \eqref{hullY} into \eqref{hullX}, we obtain
\begin{equation} \label{hullend}
G' G'^\top \;=\;
\begin{bmatrix}
Y Y^\top - I_{h-\ell} & \mathbf{0} & \mathbf{0} \\[4pt]
\mathbf{0} & \mathbf{0}_{\ell} & \mathbf{0} \\[4pt]
\mathbf{0} & \mathbf{0} & B B^\top
\end{bmatrix}.
\end{equation}
From the block form, we obtain
\[
\rank \big(G' G'^\top\big)
= \rank \big( Y Y^\top - I_{h-\ell} \big) \;+\; (k-h) = (h-\ell)+(k-h) = k-\ell,
\]
and so 
$
\dim(H(C')) = \ell,
$
as desired.
\end{proof}

\begin{example} 
We demonstrate with an example where $q=2$, $m=2$, and $n=4$. Let $\mathbb{F}_{4}=\{0,1,\omega,\omega^2\}$ with $\omega^2=\omega+1$. Let $C$ be an $[4,2]_{2^2/2}$ vector rank-metric code with  generator matrix  
\[
G=\bigl[\,I_2\ \ A\,\bigr]
=\begin{bmatrix}
1&0&\omega&\omega^2\\[2pt]
0&1&\omega^2&\omega
\end{bmatrix},\qquad
A=\begin{bmatrix}\omega&\omega^2\\[2pt]\omega^2&\omega\end{bmatrix}\in\mathbb F_4^{2\times 2}.
\]
Then
\[
A A^\top=
\begin{bmatrix}
\omega&\omega^2\\[2pt]\omega^2&\omega
\end{bmatrix}
\begin{bmatrix}
\omega&\omega^2\\[2pt]\omega^2&\omega
\end{bmatrix}^{\!\top}
=
\begin{bmatrix}
\omega^2+( \omega^2)^2 & \omega\omega^2+\omega^2\omega\\[2pt]
\omega^2\omega+\omega\omega^2 & (\omega^2)^2+\omega^2
\end{bmatrix}
=
\begin{bmatrix}
1&0\\[2pt]0&1
\end{bmatrix}
=I_2.
\]
 Hence $\dim (H(C))=2$. To apply Theorem~3.4 and reduce the hull to $\ell=0$, 
we choose from Lemma \ref{crux} the matrix
\[
Y = Z_2 =
\begin{bmatrix}1&0\\[2pt]1&1\end{bmatrix}
\in \text{GL}_2(\mathbb F_2),
\qquad
Y Y^\top - I_2 =
\begin{bmatrix}0&1\\[2pt]1&1\end{bmatrix}.
\]
Let
\[
M 
= \operatorname{diag}(Y,I_2)
= 
\begin{bmatrix}
1&0&0&0\\[2pt]
1&1&0&0\\[2pt]
0&0&1&0\\[2pt]
0&0&0&1
\end{bmatrix}\in \text{GL}_4(\mathbb F_2).
\]

\medskip
\noindent
Then 
\[
G' = G M
   = \bigl[\,I_2\ \ A\,\bigr]
     \operatorname{diag}(Y,I_2)
   = \bigl[\,Y\ \ A\,\bigr]
   = \begin{bmatrix}
       1 & 0 & \omega & \omega^2\\[2pt]
       1 & 1 & \omega^2 & \omega
     \end{bmatrix}.
\]

\medskip
\noindent
It follows that 
\[
G'G'^{\top} = Y Y^{\top} - I_2
= \begin{bmatrix}0&1\\[2pt]1&1\end{bmatrix},
\]
which is invertible over $\mathbb F_2$. 
Hence $\operatorname{rank}(G'G'^{\top})=2$ and with $C'$ the code generated by $G'$, we have
\[
\dim (H(C')) = k - \operatorname{rank}(G'G'^{\top}) = 2 - 2 = 0.
\]
\end{example}

The above construction provides a general method for varying the hull dimension of a vector rank-metric code $C$.
However, for the  case when $q \in \{2,3\}$ and $\dim(H(C))=h \ge 1$, this approach does not allow us to find an equivalent code $C'$ with $\dim(H(C'))=h-1$. We do not know a general answer, but the special case $h=1$ will be treated in the next theorem.

\begin{theorem} \label{main2} Let $q \in \{2,3\}$ and let $C$ be an $[n, k]_{q^m/q}$ vector rank-metric code with $\dim( H(C)) = 1$.  Then there exists an  $[n, k]_{q^m/q}$ vector rank-metric code $C'$  which is LCD and equivalent to $C$. 
\end{theorem}

\begin{proof}  Let $G$ be a generator matrix of $C$. Under suitable row operations on $G$, we can assume that $G$ is of the form
\[
G \;=\;
\begin{bmatrix}
1 & A \\[6pt]
 \mathbf{0}  & B
\end{bmatrix},
\]
where $A\in \mathbb{F}_{q^m}^{1\times (n-1)}$, 
and $B\in \mathbb{F}_{q^m}^{(k-1)\times (n-1)}$.
We have the following conditions
\begin{equation}\label{hullcond2}
AA^{\top} + 1 = 0, 
\qquad  
AB^{\top}=0.
\end{equation}
From the proof of Theorem \ref{main}, one can show that $GG^\top$ has the form
\[ GG^\top=
  \begin{bmatrix}
    0  & \mathbf{0}  \\
    \mathbf{0} & BB^\top
  \end{bmatrix}.
\]
Also, $\text{rank}(BB^\top)= k-1$.

1. Consider the transformation 
\[
M =
\begin{bmatrix}
1 & v \\[2pt]
 \mathbf{0} & I_{n-1}
\end{bmatrix}
\in \text{GL}_n(\mathbb{F}_q),
\qquad
v\in \mathbb{F}_q^{1\times (n-1)}.
\]
Then
\[
G' = G M =
\begin{bmatrix}
1 &  A+v\\[2pt]
 \mathbf{0} & B
\end{bmatrix}.
\]
Using the conditions \eqref{hullcond2}, the matrix $G'G'^{\top}$ has the form
\[
G'G'^{\top} =
\begin{bmatrix}
\theta &   vB^{\top}\\ 
Bv^{\top}  & BB^{\top}
\end{bmatrix},
\qquad
\]
where
$
\theta = vv^{\top} + 2vA^\top. 
$

2. We now consider the kernel of $G'G'^\top $. Let $S:=BB^{\top}$ and we note that $S$ is invertible. Solving
\[
G'G'^{\top}
\begin{pmatrix}
x\\ y
\end{pmatrix}
= 0
\]
yields $y = -S^{-1}Bv^{\top}x$, and substituting back gives
\[
(\theta - vB^{\top}S^{-1}Bv^{\top})x = 0.
\]
If $\theta - vB^{\top}S^{-1}Bv^{\top}\neq 0$, then $x=0$, and so $y=0$. This implies that $\operatorname{nullity}(G'G'^{\top}) = 0$. Consequently $\operatorname{rank}(G'G'^{\top})=n$ and $\dim H(C')=0$.
In the case $\theta - vB^{\top}S^{-1}Bv^{\top}=0$, we obtain  $\dim (H(C'))=1$ instead. 
Therefore,
\[
\dim (H(C')) = 
\begin{cases}
0, & \text{if } \theta - vB^{\top}S^{-1}Bv^{\top} \neq 0,\\[4pt]
1,   & \text{if } \theta - vB^{\top}S^{-1}Bv^{\top} = 0.
\end{cases}
\]
3. Let $P:=B^{\top}S^{-1}B=B^{\top}(BB^\top)^{-1}B$.  
 Then 
\[
\theta - vB^{\top}S^{-1}Bv^{\top}  =  vv^{\top} + 2vA^{\top}- vPv^{\top} = v(I_{n-1} - P)v^\top +2vA^{\top}.
\]
Let $Q:=I_{n-1}-P$. It is easy to check that $Q$ is symmetric.
Here, we note that $\rank(P)=\rank(B)=k-1$. This implies that  the matrix $Q$ is nonzero.
Define
\[
f(v) :=  vQv^\top +2vA^{\top}.
\]
In the remainder of the proof, we construct a vector $v   \in\F_q^{\,n-1}$ such that $f(v) \ne 0$. 
There are two cases depends on  $q$.

% (i) \underline{$q$ is even}. In this case,  $f(v)$ becomes the quadratic form
% \[
% f(v) =  vQv^\top 
% \]
% over $\mathbb{F}_{q^m}$.  For any \(v=(v_1,\dots,v_{n-1})\in\F_q^{\,n-1}\), we have
% \[
% vQv^{\top}=\sum_{i=1}^{n-1} q_{ii}\,v_i^{\,2}.
% \]
%  Since \(Q\) is idempotent,  
% \[
% \operatorname{tr}(Q)=\rank(Q)=n-k>0.
% \]
% Therefore some diagonal entry \(q_{ii}\) is nonzero. Choosing
% \(v=e_i\in\F_q^{\,n-1}\) gives
% \[
% f(v)  = vQv^{\top}=q_{ii}\neq 0.
% \]

% (ii) \underline{$q$ is odd}. 
 
\underline{Case $q=2$}. In this case, $f(v)$ becomes the nontrivial quadratic form
\[
f(v) = v Q v^\top
\]
over $\mathbb{F}_{q^m}$. For any 
$
x = (x_1, x_2, \ldots, x_{n-1}) \in \mathbb{F}_{q^m}^{\,n-1},
$
we have
\[
x Q x^\top = \sum_{i=1}^{n-1} q_{ii} x_i^2.
\]
 Since $ A B^\top = 0$, we have
\[
A P A^\top 
 = A\,B^\top S^{-1} B\,A^\top 
 = (A B^\top)\,S^{-1}\,(B A^\top) 
 = 0,
\]
and so
\[
A Q A^\top 
= A(I - P)A^\top 
= A A^\top - A P A^\top 
= A A^\top 
= 1.
\] 
The condition $A Q A^\top =1$ implies that some diagonal entry $q_{ii}$ is nonzero. Choosing $v = e_i \in \mathbb{F}_q^{\,n-1}$ gives
\[
f(v) = v Q v^\top = q_{ii} \neq 0.
\]

 \underline{Case $q=3$}. The condition $AA^\top+1=0$ implies that $A \ne \mathbf{0}.$ Let $v \in \mathbb{F}_q^{n-1}$ be such that $vA^{\top}\neq 0$.
If $f(v)\neq 0$, then we are done. Otherwise $f(v)=0$ implies
\[
2\,vA^{\top}=-vQv^{\top}\neq 0.
\]
For any $\alpha\in\F_q$,
\[
f(\alpha v)=\alpha^2\,vQv^{\top}+2\alpha v A^{\top}
=\alpha^2\,vQv^{\top}-\alpha vQv^{\top}
=\alpha(\alpha-1)\,(vQv^{\top}).
\]
Choose $\alpha=-1\in\F_q\setminus\{0,1\}$, then the vector $\alpha v$ is in $\mathbb{F}_q^{n-1}$ and satisfies $f(\alpha v)\neq 0$.

In all cases, there exists $v$ such that $f(v)\neq 0$, and from part 2 of the proof, we obtain $\dim(H(C')) = 0$. 
\end{proof}

\begin{example}  
Let $q=2$ and $m=2$. Let $\mathbb{F}_4 = \{0,1,\omega,\omega^2\}$ with $\omega^2 = \omega + 1$.
Consider the $[4,2]_{2^2/2}$ vector rank-metric code $C$ with dimension $k=2$ generated by
\[
G =
\begin{bmatrix}
1 & \omega & 0 & \omega^2\\[4pt]
0 & \omega & 0 & 1
\end{bmatrix}.
\]
1. We  rewrite
\[
G =
\begin{bmatrix}
1 & A\\[2pt]
0 & B
\end{bmatrix}, \qquad
A = [\,\omega\ \ 0\ \ \omega^2\,], \quad
B = [\,\omega\ \ 0\ \ 1\,].
\] 
Here, $AA^{\top} + 1 = 0$, $AB^{\top} = 0$,
\[
GG^{\top} =
\begin{bmatrix}
0 & 0\\[4pt]
0 & BB^{\top}
\end{bmatrix}, \qquad
BB^{\top} = \omega \neq 0.
\]
Hence,
\[
\operatorname{rank}(GG^{\top}) = 1, \qquad \dim (H(C)) = k - \operatorname{rank}(GG^{\top}) = 2 - 1 = 1.
\]
2.  Let 
\[
M =
\begin{bmatrix}
1 & v\\[2pt]
0 & I_3
\end{bmatrix}
=
\begin{bmatrix}
1 & 1 & 0 & 0\\[2pt]
0 & 1 & 0 & 0\\[2pt]
0 & 0 & 1 & 0\\[2pt]
0 & 0 & 0 & 1
\end{bmatrix}, \qquad
v = [\,1\ 0\ 0\,] \in \mathbb{F}_2^{1\times 3}.
\]
Let $S = BB^{\top} =\omega.
$
Then $S^{-1} = \omega^2$. We then compute
\[
P = B^{\top}\omega^2 B =
\begin{bmatrix}
\omega & 0 & 1\\[2pt]
0 & 0 & 0\\[2pt]
1 & 0 & \omega^2
\end{bmatrix}, \qquad
Q = I_3 - P =
\begin{bmatrix}
1+\omega & 0 & 1\\[2pt]
0 & 1 & 0\\[2pt]
1 & 0 & 1+\omega^2
\end{bmatrix}.
\]
 From Theorem \ref{main2}, when $q = 2$, we define
$
f(v) = vQv^{\top}.
$
Note that each diagonal entry of $Q$ is nonzero in this case. We simply choose $v=e_1$, so that 
$
f(v) = vQv^{\top} 
 = 1 + \omega \neq 0.
$

3. 
The new generator matrix is
\[
G' = G\,M =
\begin{bmatrix}
1 & \omega^2 & 0 & \omega^2\\[4pt]
0 & \omega & 0 & 1
\end{bmatrix}.
\] 
Compute
\[
G'G'^{\top} =
\begin{bmatrix}
1 & \omega\\[4pt]
\omega & \omega
\end{bmatrix},
\qquad
\det(G'G'^{\top}) = \omega + \omega^2 = 1 \neq 0.
\] 
Hence
$
\operatorname{rank}(G'G'^{\top}) = 2.
$ 
It follows that the   code $C'$ defined by $G'$ is LCD and equivalent to $C$. 
\end{example}

From Theorems \ref{main} and \ref{main2}, we obtain the following corollary.

\begin{corollary} Let $q$ be a prime power and let $C$ be an $[n, k]_{q^m/q}$ vector rank-metric code.  Then there exists an $[n, k]_{q^m/q}$ vector rank-metric code $C'$ which is LCD and equivalent to $C$.
\end{corollary}

\section{On the hull of codes associated with vector rank-metric codes}

In this section, we assume the additional condition that $q$ is even or  both \( q \) and \( m \) are odd. Then there exists a self-dual basis   of $\mathbb{F}_{q^m}$ over $\mathbb{F}_q$, see Subsection 2.6. 
From Theorem~\ref{main}, we have the following corollary for matrix codes  associated with vector rank-metric codes. 
 
\begin{corollary} \label{cor1} Let $q$ be  even or both \( q \) and \( m \) are odd. Let $\mathcal{G}$ be a self-dual basis of $\mathbb{F}_{q^m}$ over $\mathbb{F}_q$.
Let $C$ be an $[n, k]_{q^m/q}$ vector rank-metric code with $\dim_{\mathbb{F}_{q^m}}( H(C)) = h \ge 1$. 
Let $D:=\mathcal{C}_\mathcal{G}(C)$ be the matrix code associated with $C$. 

\begin{enumerate} 
\item If $q>3$, then there exists a matrix code $D'$ equivalent to $D$ with
 $\dim_{\mathbb{F}_{q}}(H(D')) = m\ell$ for each $\ell \in \{0,1,\dots, h-1\}$. 

\item If $q \in \{2,3\}$ and $h \ge 2$, then there exists a matrix code $D'$ equivalent to $D$ with
 $\dim_{\mathbb{F}_{q}}(H(D')) = m\ell$ for each $\ell \in \{0,1,\dots, h-2\}$. 

\item If $q \in \{2,3\}$ and $h =1$, then there exists a matrix code $D'$ which is LCD and equivalent to $D$.  
 \end{enumerate} 
\end{corollary}
 
 \begin{proof} Fix $\ell$ as follows: 
$\ell \in \{0, \dots, h-1\}$ if $q > 3$; 
$\ell \in \{0, \dots, h-2\}$ if $q \in \{2,3\}$ and $h \ge 2$; 
and $\ell = 0$ if $q \in \{2,3\}$ and $h = 1$. 
By Theorems \ref{main} and \ref{main2}, there exists a vector rank-metric code $C'$ equivalent to $C$ with $\dim_{\mathbb{F}_{q^m}}(H(C'))=\ell$. 
Let $D':=\mathcal{C}_\mathcal{G}(C')$ be the matrix code associated with $C'$.
Then by Lemma \ref{hulltransfer}, we have that  
\[
H(D') = \mathcal{C}_\mathcal{G}(C') \cap \mathcal{C}_\mathcal{G}(C')^\perp=\mathcal{C}_\mathcal{G}(C \cap C^\perp) = \mathcal{C}_\mathcal{G}(H(C')),
\]
and so
\[
 \dim_{\mathbb{F}_{q}}(H(D')) = m \cdot \dim_{\mathbb{F}_{q^m}} (H(C')) = ml.
\]
Furthermore, $D'$ is equivalent to $D$  by Lemma \ref{equivalencetransfer}, and this completes the proof.
\end{proof}
 
% \subsection{Extended block codes associated with vector rank-metric codes}  
% \begin{definition}

% We denote the collection of linear $l \times m$ matrix-equivalence maps on extended row vectors by $Equiv_{\text{Vec}}(\mathbb{F}_q^{l \times m})$.
% \end{definition}

We now consider the hull-variation problem for extended block codes associated with vector rank-metric codes. 
\begin{lemma} \label{hulltransferblock} Let $C$ and $C'$ be two matrix codes in $\mathbb{F}_q^{n \times m}$. Then $\rho(C \cap C') = \rho(C) \cap \rho(C')$. 
\end{lemma}
\begin{proof} It can be readily checked that $\rho(C \cap C') \subseteq \rho(C) \cap \rho(C')$. For the converse, let $c \in \rho(C) \cap \rho(C')$. Then $\rho^{-1}(c) \in C$ and  $\rho^{-1}(c) \in  C'$, so that $ \rho^{-1}(c) \in C \cap C'$. Then, $c \in   \rho(C \cap C')$.  
\end{proof}

 \begin{corollary}Let $q$  be  even or both \( q \) and \( m \) are odd. Let $C$ be an $[n, k]_{q^m/q}$ vector rank-metric code with $\dim_{\mathbb{F}_{q^m}}( H(C)) = h$. Let $D:=\rho(\mathcal{C}_\mathcal{G}(C))$ be the extended block code  associated with $C$. 
 
 \begin{enumerate} 
\item If $q>3$, then there exists an extended block code $D'$ which is linearly $n \times m$ matrix-equivalent to $D$ with
 $\dim_{\mathbb{F}_{q}}(H(D')) = m\ell$ for each $\ell \in \{0,1,\dots, h-1\}$. 

\item If $q \in \{2,3\}$ and $h \ge 2$, there exists an extended block code $D'$ which is linearly $n \times m$ matrix-equivalent to $D$ with
 $\dim_{\mathbb{F}_{q}}(H(D')) = m\ell$ for each $\ell \in \{0,1,\dots, h-2\}$. 

 \item If $q \in \{2,3\}$ and $h =1$, then there exists an extended block code $D'$ which is LCD and linearly $n \times m$ matrix-equivalent to $D$.
 \end{enumerate}

 % Then there exists an extended block code $D'$ which is linearly $n \times m$ matrix-equivalent to $D$ with
 % $\dim_{\mathbb{F}_{q}}(H(D')) = ml$ for each $l \in \{0,1,\dots, h\}$. 
 \end{corollary}

 \begin{proof} Fix $\ell$ as follows: 
$\ell \in \{0, \dots, h-1\}$ if $q > 3$; 
$\ell \in \{0, \dots, h-2\}$ if $q \in \{2,3\}$ and $h \ge 2$; 
and $\ell = 0$ if $q \in \{2,3\}$ and $h = 1$. 
 By Theorems \ref{main} and \ref{main2}, there exists a vector rank-metric code $C'$ equivalent to $C$ with $\dim_{\mathbb{F}_{q^m}}(H(C'))=\ell$. Let $D':=\rho(\mathcal{C}_\mathcal{G}(C'))$ be the extended block code   associated with $C'$. By Lemma \ref{equivalencetransfer}, the matrix codes $\mathcal{C}_\mathcal{G}(C)$ and $\mathcal{C}_\mathcal{G}(C')$ are equivalent. Then by Lemma \ref{equivalencetransferblock},   $D'$   is linearly $n \times m$ matrix-equivalent to $D$.
 On the other hand, 
\begin{align*}
H(D') = D'\cap D'^\perp &= \rho(\mathcal{C}_\mathcal{G}(C')) \cap \rho(\mathcal{C}_\mathcal{G}(C'))^\perp \\
&= \rho(\mathcal{C}_\mathcal{G}(C')) \cap \rho(\mathcal{C}_\mathcal{G}(C')^\perp) \tag{by   \eqref{rhocom}}\\
&= \rho(\mathcal{C}_\mathcal{G}(C')) \cap \rho(\mathcal{C}_\mathcal{G}(C'^\perp)) \tag{by Lemma \ref{ravagnani}}\\
&= \rho(\mathcal{C}_\mathcal{G}(C') \cap \mathcal{C}_\mathcal{G}(C'^\perp)) \tag{by Lemma \ref{hulltransferblock}} \\
&= \rho(\mathcal{C}_\mathcal{G}(C' \cap C'^\perp)) \tag{by Lemma \ref{hulltransfer}} \\
&= \rho(\mathcal{C}_\mathcal{G}(H(C'))).
\end{align*}
 From the proof of Corollary \ref{cor1}, we have that $\dim_{\mathbb{F}_q}( \mathcal{C}_\mathcal{G}(H(C'))) = m\ell$. It follows that  $\dim_{\mathbb{F}_q} (H(D')) = m\ell,$ which completes the proof.
 \end{proof}

 \begin{remark}
An interesting question concerns $\mathbb{F}_q$-linear rank-metric codes that are not obtained as associated codes from vector rank-metric codes. 
The behavior of their hulls under equivalence is, to the best of our knowledge, not yet fully understood and may be considered as an open problem.
\end{remark}

\begin{remark} In this final remark, we consider the implications of our results in the context of matroid theory. 
For binary and ternary Hamming-metric codes, it is well-known that the hull dimension is completely determined by the associated matroid 
and can be expressed in terms of its Tutte polynomial, see  \cite{BO} and \cite{J}.
Thus, equivalent Hamming-metric codes over $\mathbb{F}_2$ or $\mathbb{F}_3$ always share the same hull dimension. 
In contrast, for vector rank-metric and matrix codes, equivalent codes may induce isomorphic $q$-matroids or $(q,m)$-polymatroids (see for example \cite{gorla2019}, \cite{jurrius2018} and \cite{shiromoto2019}),
while having different hull dimensions. 
In other words, in the rank-metric setting, the hull dimension is not an invariant captured by the underlying matroid-like structures. 
\end{remark}

\section*{Acknowledgment} The authors wish to thank the Co-Editor-in-Chief Sihem Mesnager and the anonymous referees for their comments   which improved the presentation and quality of this paper. 
D. Ho and T. Johnsen were supported by  the Tromsø Research Foundation (project “Pure Mathematics in Norway”) and  UiT Aurora project MASCOT.

\end{document}